\documentclass[copyright,creativecommons]{eptcs}

\pdfoutput=1

\hypersetup{
    pdftitle={%
        Proofs about Network Communication:
        For Humans and Machines%
    },
    pdfauthor={Wolfgang Jeltsch, Javier Díaz}
}

\providecommand{\event}{ICE 2023}

\title{
    Proofs about Network Communication:\\
    For Humans and Machines
}
\newcommand{\titlerunning}{%
    Proofs about Network Communication:
    For Humans and Machines%
}
\author{
    Wolfgang Jeltsch
    \institute{Well-Typed\\London, England}
    \email{wolfgang@well-typed.com}
\and
    Javier Díaz
    \institute{Atix Labs (a Globant Division)\\Buenos Aires, Argentina}
    \email{javier.diaz@globant.com}
}
\newcommand{\authorrunning}{Wolfgang Jeltsch \& Javier Díaz}
\newcommand{\copyrightholders}{Input Output}

\usepackage{amsmath}
\usepackage{amsthm}

\newtheorem{definition}{Definition}
\newtheorem{example}{Example}
\newtheorem{proposition}{Proposition}
\newtheorem{lemma}{Lemma}
\newtheorem{corollary}{Corollary}

\usepackage[T1]{fontenc}

\usepackage{newtxtext}
\usepackage{newtxmath}

\usepackage[all]{hypcap}

\usepackage{listings}

\lstdefinelanguage{isabelle}{
    keywords={
        lemma,
        shows,
        proof,qed,
        next,
        case,
        then,moreover,ultimately,
        from,with,
        have,show,obtain,
        where,
        and,for
    }
}
\newcommand{\conj}{\mathrel\wedge}
\newcommand{\uniquant}[1]{\forall#1.\:}
\newcommand{\app}{\,}

\newcommand{\bnfdef}{\mathrel{{\mathop:}{\mathop:}{=}}}

\newcommand{\deadlock}{\mathbf{0}} 
\newcommand{\send}{\lhd}
\newcommand{\receive}[1]{\rhd#1.\:}
\newcommand{\newchannel}[1]{\nu#1.\:}

\newcommand{\repeatedreceive}[1]{\rhd^{\infty}#1.\:}

\newcommand{\distributor}{\mathbin\Rightarrow}
\newcommand{\unidirectionalbridge}{\mathbin\rightarrow}
\newcommand{\bidirectionalbridge}{\mathbin\leftrightarrow}
\newcommand{\unreliability}[1]{\text{\textup{\textcurrency}}^{#1}}
\newcommand{\loser}{\unreliability{?}}
\newcommand{\duplicator}{\unreliability{+}}
\newcommand{\duploser}{\unreliability{*}}

\newcommand{\sending}{\lhd}
\newcommand{\receiving}{\rhd}

\newcommand{\trans}[1]{\xrightarrow{\,#1\,}}

\newenvironment{rulesinfigure}
    {\begin{gathered}}
    {\vspace{0.5\topsep}\end{gathered}}
\newcommand{\introrule}[3]{\frac{#1}{#2}\quad(#3)}
\newcommand{\nextrule}{\qquad\qquad}
\newcommand{\nextline}{\\[\topsep]}

\newcommand{\newrulename}[2]
    {\expandafter\newcommand\csname rule#1\endcsname{{#2}}}

\newrulename{sending}{\lhd}
\newrulename{receiving}{\rhd}
\newrulename{communicationltr}{\tau_{\rightarrow}}
\newrulename{communicationrtl}{\tau_{\leftarrow}}
\newrulename{parallelleft}{\parallel_{1}}
\newrulename{parallelright}{\parallel_{2}}
\newrulename{newchannel}{\nu}

\newrulename{distribution}{\Rightarrow}

\newcommand{\sendingrule}
    {
        \introrule
            {}
            {a \send x \trans{a \sending x} \deadlock}
            {\rulesending}
    }
\newcommand{\parallelandnewchannelrules}
    {
        \introrule
            {
                p \trans{a \sending x} p'
                \quad
                q \trans{a \receiving x} q'
            }
            {p \parallel q \trans{\tau} p' \parallel q'}
            {\rulecommunicationltr}
        \nextrule
        \introrule
            {
                p \trans{a \receiving x} p'
                \quad
                q \trans{a \sending x} q'
            }
            {p \parallel q \trans{\tau} p' \parallel q'}
            {\rulecommunicationrtl}
        \nextline
        \introrule
            {p \trans{\alpha} p'}
            {p \parallel q \trans{\alpha} p' \parallel q}
            {\ruleparallelleft}
        \nextrule
        \introrule
            {q \trans{\alpha} q'}
            {p \parallel q \trans{\alpha} p \parallel q'}
            {\ruleparallelright}
        \nextrule
        \introrule
            {\uniquant{a} P \app a \trans{\alpha} Q \app a}
            {
                \newchannel{a} P \app a
                \trans{\alpha}
                \newchannel{a} Q \app a
            }
            {\rulenewchannel}
    }

\newcommand{\statementref}[1]{(\ref*{#1})}

\newcommand{\fact}[1]{\text{\guilsinglleft}#1\text{\guilsinglright}}
\newcommand{\chain}{\mathbin\frown}
\newcommand{\constantbisimilarity}{[\sim]}
\newcommand{\mutantlifting}{\mathcal{M}}
\newcommand{\schematic}[1]{\textit{?}\mathit{#1}}
\newcommand{\IO}[4]
    {\mathit{IO} \app #1 \app #2 \app #3 \app #4}
\newcommand{\postreceive}[3]
    {\mathit{post\_receive} \app #1 \app #2 \app #3}
\newcommand{\suffixadapted}
    {\mathbin\text{\guillemetright} \mathit{suffix} \app}
\newcommand{\realreceiving}[1]{\receiving_{#1}}
\newcommand{\proofplaceholder}
    {\langle\mathit{proof}\rangle}

\newcommand{\nobinarg}{{}\mskip-\medmuskip}

\begin{document}

\maketitle

\begin{abstract}

Many concurrent and distributed systems are safety-critical and
therefore have to provide a high degree of assurance. Important
properties of such systems are frequently proved on the specification
level, but implementations typically deviate from specifications for
practical reasons. Machine-checked proofs of bisimilarity statements are
often useful for guaranteeing that properties of specifications carry
over to implementations. In this paper, we present a way of conducting
such proofs with a focus on network communication. The proofs resulting
from our approach are not just machine-checked but also intelligible for
humans.

\end{abstract}

\section{Introduction}

\label{introduction}

Concurrent and distributed systems are difficult to design and
implement, and their correctness is hard to ensure. However, many such
systems are safety-critical and therefore have to provide a high degree
of assurance. Machine-checked proofs can greatly help to meet this
demand. A particular application area of them is the verification of
design refinements. A specification may undergo a series of refinement
steps to account for practical limitations, ultimately resulting in an
implementation. Proving that these refinement steps preserve important
properties of the system is vital for assuring the implementation’s
correctness.

Our current research program focuses on applying design refinement
verification to the blockchain consensus protocols of the Ouroboros
family~\cite{badertscher:2018,david:2018,kiayias:2017}. For conducting
machine-checked proofs, we use the Isabelle proof assistant together
with a custom process calculus, called the Þ-calculus. As a first step,
we have proved~\cite{jeltsch:2022} that direct broadcast, which the
protocol specifications assume as the means of data distribution, is
behaviorally equivalent to broadcast via multicast, which
implementations of the protocols use. For our proof, we have used a
domain-specific language for describing network communication, which is
embedded in the Þ-calculus.

A weakness of this existing broadcast equivalence proof is that it is
not grounded in a formal semantics of the communication language but
based on the assumption that certain lower-level bisimilarity statements
hold. In this paper, we present a way of proving such bisimilarity
statements such that the resulting proofs are machine-checked and
intelligible. Concretely, we make the following contributions:
\begin{itemize}

\item

We present a transition system semantics for the Þ-calculus and derive a
transition system semantics for the communication language from it.

\item

We walk in detail through the proof of a lemma from which several
fundamental bisimilarity statements about communication language
processes can be derived. The proof of this lemma exemplifies a general
way of conducting bisimulation proofs in a concise and human-friendly
yet machine-checked fashion. Central to this approach is the combination
of the Isabelle/Isar proof language, a formalized algebra of “up to”
methods, Isabelle’s \lstinline|coinduction| proof method, and
higher-order abstract syntax.

\end{itemize}

The formal broadcast equivalence proof and its prerequisites can be
obtained from the following sources:
\begin{itemize}

\item

\url{https://github.com/input-output-hk/equivalence-reasoner}

\item

\url{https://github.com/input-output-hk/transition-systems}

\item

\url{https://github.com/input-output-hk/thorn-calculus}

\item

\url{https://github.com/input-output-hk/network-equivalences}

\end{itemize}

\section{The Þ-Calculus}

The Þ-calculus (pronounced “thorn calculus”) is a general-purpose
process calculus, which we have devised as a tool for convenient
development of machine-checked proofs about concurrent and distributed
systems. Our language for describing communication networks is embedded
in the Þ-calculus, so that we can leverage the strengths of the
Þ-calculus in our network-related specifications and proofs. The
Þ-calculus in turn is embedded in Isabelle/HOL. We use higher-order
abstract syntax (HOAS) for this embedding, since this allows us to have
the object language (the Þ-calculus) only deal with the key features of
process calculi, which are concurrency and communication, while shifting
the treatment of local names, data, computation, conditional execution,
and repetition to the meta-language (Isabelle/HOL).

The Þ-calculus strongly resembles the asynchronous
$\pi$-calculus~\cite{honda:1991}. Processes communicate via asynchronous
channels, which can be global or created locally. Channels are
first-class and can therefore be transmitted through other channels,
thus making them visible outside their original scopes. This is the
mobility feature pioneered by the (synchronous) $\pi$-calculus. However,
mobility plays only a marginal role in this paper, since it is not
exploited by the communication language.

\begin{definition}[Syntax of Þ-calculus processes]

The syntax of Þ-calculus processes is given by the following BNF rule,
where $a$~denotes channels, $x$~denotes values, $p$~and~$q$ denote
processes, and $P$~denotes functions from channels or from values to
processes, depending on the context:\footnote{Note that we write
function applications as mere juxtapositions of the respective functions
and their arguments, which is in line with Isabelle notation.}
\begin{equation*}
\mathit{Process} \bnfdef
    \deadlock               \mid
    a \send x               \mid
    a \receive{x} P \app x  \mid
    p \parallel q           \mid
    \newchannel{a} P \app a
\end{equation*}
The processes generated by the different alternatives of this BNF rule
are called the stop process, senders, receivers, parallel compositions,
and restrictions, respectively. Parallel composition is
right-associative and has lowest precedence; the other constructs have
highest precedence.

\end{definition}

Informally, the semantics of the Þ-calculus is characterized by the
following behavior of processes:
\begin{itemize}

\item

The stop process~$\deadlock$ does nothing.

\item

A sender $a \send x$ sends value~$x$ to channel~$a$.

\item

A receiver $a \receive{x} P \app x$ receives a value~$x$ from
channel~$a$ and continues like~$P \app x$.

\item

A parallel composition $p \parallel q$ performs $p$~and~$q$ in parallel.

\item

A restriction $\newchannel{a} P \app a$ introduces a local channel~$a$
and behaves like~$P \app a$.

\end{itemize}

Formally, the semantics is defined as a labeled transition system. Since
mobility is not essential to the topics of this paper and is at the same
time complex to handle, we present only a simplified version of the
semantics that ignores mobility.\footnote{We refer the reader to the
accompanying Isabelle code for the full semantics.}

\begin{definition}[Syntax of Þ-calculus actions]

The syntax of Þ-calculus actions is given by the following BNF rule,
where $a$~denotes channels and $x$~denotes values:
\begin{equation*}
\mathit{Action} \bnfdef
    a \sending x   \mid
    a \receiving x \mid
    \tau
\end{equation*}
The actions generated by the different alternatives of this BNF rule are
called sending actions, receiving actions, and the internal-transfer
action, respectively.

\end{definition}

The intuitive meanings of the different actions are as follows:
\begin{itemize}

\item

A sending action $a \sending x$ means sending value~$x$ to channel~$a$.

\item

A receiving action $a \receiving x$ means receiving value~$x$ from
channel~$a$.

\item

The internal-transfer action~$\tau$ means transferring some value
through some channel.

\end{itemize}

\begin{definition}[Semantics of the Þ-calculus]

The semantics of the Þ-calculus is given by the transition relation
${\trans{}} \subseteq \mathit{Process} \times \mathit{Action} \times
\mathit{Process}$ that is defined by the introduction rules in
\autoref{thorn-calculus-no-mobility-transition-rules}.%
\begin{figure}

\centering
\begin{math}
\begin{rulesinfigure}
\sendingrule
\nextrule
\introrule
    {}
    {a \receive{x} P \app x \trans{a \receiving x} P \app x}
    {\rulereceiving}
\nextline
\parallelandnewchannelrules
\end{rulesinfigure}
\end{math}

\caption{No-mobility versions of the transition rules of the Þ-calculus}

\label{thorn-calculus-no-mobility-transition-rules}

\end{figure}

\end{definition}

As usual, $p \trans{\alpha} q$ intuitively means that process~$p$ can
perform action~$\alpha$ and then continue like process~$q$.

\begin{definition}[Strong and weak bisimilarity of Þ-calculus processes]

The relations ${\sim} \subseteq \mathit{Process} \times
\mathit{Process}$ and ${\approx} \subseteq \mathit{Process} \times
\mathit{Process}$ denote strong and weak bisimilarity derived
from~$\trans{}$ in the usual way.

\end{definition}

Strong and weak bisimilarity possess various properties common for
process calculi, for which proofs can be found in the accompanying
Isabelle code.

\begin{lemma}[Inclusion of strong bisimilarity in weak bisimilarity]

\label{strong-weak-bisimilarity-inclusion}

Strongly bisimilar processes are also weakly bisimilar; formally,
${\sim} \subseteq {\approx}$.

\end{lemma}

\begin{lemma}[Congruence properties of bisimilarities]

Strong and weak bisimilarity are congruence relations with respect to
parallel composition and restriction; that is, they are equivalence
relations, and the following propositions hold:
\begin{align}
\label{parallel-compatibility-strong}
p_{1} \sim p_{2} \conj q_{1} \sim q_{2}
& \rightarrow
p_{1} \parallel q_{1} \sim p_{2} \parallel q_{2}
\\
\label{parallel-compatibility-weak}
p_{1} \approx p_{2} \conj q_{1} \approx q_{2}
& \rightarrow
p_{1} \parallel q_{1} \approx p_{2} \parallel q_{2}
\\
(\uniquant{a} P_{1} \app a \sim P_{2} \app a)
& \rightarrow
\newchannel{a} P_{1} \app a \sim \newchannel{a} P_{2} \app a
\\
(\uniquant{a} P_{1} \app a \approx P_{2} \app a)
& \rightarrow
\newchannel{a} P_{1} \app a \approx \newchannel{a} P_{2} \app a
\end{align}

\end{lemma}

\begin{lemma}[Fundamental bisimilarity properties]

\label{fundamental-bisimilarities}

The following strong bisimilarity properties hold:
\begin{align}
\label{parallel-left-identity}
\deadlock \parallel p
& \sim
p
\\
p \parallel \deadlock
& \sim
p
\\
\label{parallel-associativity}
(p \parallel q) \parallel r
& \sim
p \parallel (q \parallel r)
\\
\label{parallel-commutativity}
p \parallel q
& \sim
q \parallel p
\\
\newchannel{a} \newchannel{b} P \app a \app b
& \sim
\newchannel{b} \newchannel{a} P \app a \app b
\\
\newchannel{a} p
& \sim
p
\end{align}

\end{lemma}

Unlike the asynchronous $\pi$-calculus, the Þ-calculus does not contain
a construct for guarded recursion, and it also does not contain a
replication construct as found in the synchronous $\pi$-calculus. This
is because the use of HOAS allows us to resort to the recursion features
of the host language and in particular to build infinite processes,
since the type of processes is coinductive. We could use this
possibility to define guarded recursion and replication on top of the
Þ-calculus and also directly to construct processes involving
repetition. However, we introduce a guarded replication construct
instead, which we use to realize any repetition.

\begin{definition}[Repeating receivers]

Processes $a \repeatedreceive{x} P \app x$, where $a$~denotes channels,
$x$~denotes values, and $P$~denotes functions from values to processes,
are defined as follows:
\begin{equation}
\label{repeating-receivers}
a \repeatedreceive{x} P \app x
=
a \receive{x} (P \app x \parallel a \repeatedreceive{x} P \app x)
\end{equation}
Such processes are called repeating receivers. The precedence of
$\rhd^{\infty}$ is the same as the one of~$\rhd$.

\end{definition}

As can be seen from \autoref{repeating-receivers}, a repeating receiver
$a \repeatedreceive{x} P \app x$ repeatedly receives values~$x$ from
channel~$a$ and after each receipt initiates the execution of $P \app
x$.

\begin{lemma}[Transitions from repeating receivers]

\label{repeating-receivers-transitions}

The only transitions possible from repeating receivers are of the form
$a \repeatedreceive{x} P \app x \trans{a \receiving x} P \app x
\parallel a \repeatedreceive{x} P \app x$.

\end{lemma}

\begin{proof}

According to \autoref{repeating-receivers}, repeating receivers are
receivers of a special kind. The only transition rule that introduces
transitions from receivers is~$\rulereceiving$, and applying this rule
to repeating receivers leads to transitions of the above-mentioned form.
\end{proof}

\section{The Communication Language}

The communication language is a process calculus specifically designed
for describing communication networks. It differs from the Þ-calculus in
that it does not allow for arbitrary sending and receiving but instead
provides constructs for forwarding, removing, and duplicating values in
channels. These constructs are more high-level than the Þ-calculus
constructs they replace. They are also more limiting but still permit
the communication language to express data flow in a network. By staying
within the confines of the communication language, our network-related
specifications and proofs tend to be well structured and comprehensible.

\begin{definition}[Syntax of communication language processes]

The syntax of communication language processes is given by the following
BNF rule, where $a$,~$b$, and~$b_{i}$ denote channels, $p$~and~$q$
denote processes, and $P$~denotes functions from channels to processes:
\begin{equation*}
\mathit{Process} \bnfdef
    \deadlock                             \mid
    a \distributor [b_{1}, \ldots, b_{n}] \mid
    a \unidirectionalbridge b             \mid
    a \bidirectionalbridge b              \mid
    \loser a                              \mid
    \duplicator a                         \mid
    \duploser a                           \mid
    p \parallel q                         \mid
    \newchannel{a} P \app a
\end{equation*}
The processes generated by the different alternatives of this BNF rule
are called the stop process, distributors, unidirectional bridges,
bidirectional bridges, losers, duplicators, duplosers, parallel
compositions, and restrictions, respectively. Parallel composition is
right-associative and has lowest precedence; the other constructs have
highest precedence.

\end{definition}

The stop process, parallel compositions, and restrictions behave like
they do in the Þ-calculus. The behavior of the other communication
language constructs is informally characterized as follows:
\begin{itemize}

\item

A distributor $a \distributor [b_{1}, \ldots, b_{n}]$ continuously
forwards values from channel~$a$ to all channels~$b_{i}$.

\item

A unidirectional bridge $a \unidirectionalbridge b$ continuously
forwards values from channel~$a$ to channel~$b$.

\item

A bidirectional bridge $a \bidirectionalbridge b$ continuously forwards
values from channel~$a$ to channel~$b$ and from channel~$b$ to
channel~$a$.

\item

A loser $\loser a$ continuously removes values from channel~$a$.

\item

A duplicator $\duplicator a$ continuously duplicates values in
channel~$a$.

\item

A duploser $\duploser a$ continuously removes values from and duplicates
values in channel~$a$.

\end{itemize}

\begin{example}[Reliable anycast with three receivers]

\label{reliable-anycast}

Consider a reliable anycast connection between a sender and three
receivers, the latter being numbered from~$1$ to~$3$. Assume that the
sender is equipped with a buffer for packets to be sent and each
receiver is equipped with a buffer for packets received. If we model the
sender’s buffer by a channel~$s$ and the buffer of each receiver~$i$ by
a channel~$r_{i}$, this anycast connection can be modeled by the
following process:
\begin{equation*}
\newchannel{t}
(
    s \unidirectionalbridge t     \parallel
    t \unidirectionalbridge r_{1} \parallel
    t \unidirectionalbridge r_{2} \parallel
    t \unidirectionalbridge r_{3}
)
\end{equation*}
Note that values in the local channel~$t$ model packets in transit.

\end{example}

\begin{example}[Unreliable broadcast with three receivers]

\label{unreliable-broadcast}

Consider a broadcast connection between a sender and three receivers
that is unreliable in the sense that packets may be lost or duplicated.
If channels $s$,~$r_{1}$, $r_{2}$, and~$r_{3}$ model send and receive
buffers like in \autoref{reliable-anycast}, this broadcast connection
can be modeled by the following process:
\begin{equation*}
\newchannel{t}
(
    s \unidirectionalbridge t     \parallel
    \duploser t                   \parallel
    t \unidirectionalbridge r_{1} \parallel
    t \unidirectionalbridge r_{2} \parallel
    t \unidirectionalbridge r_{3}
)
\end{equation*}
This process models indeed a broadcast connection, not an anycast
connection, because due to duplication a single value sent to~$t$ may be
forwarded to different channels~$r_{i}$.

\end{example}

\begin{definition}
    [Embedding of the communication language in the Þ-calculus]

\label{communication-language-embedding}

The communication language is a DSL embedded in the Þ-calculus. The stop
process, parallel composition, and restriction are directly taken from
the Þ-calculus, and the other communication language constructs are
derived from Þ-calculus constructs and the repeating receiver construct
as follows:
\begin{align}
a \distributor [b_{1}, \ldots, b_{n}]
& =
a \repeatedreceive{x}
(
    b_{1} \send x \parallel
    \ldots        \parallel
    b_{n} \send x \parallel
    \deadlock
)
\\
a \unidirectionalbridge b
& =
a \distributor [b]
\\
a \bidirectionalbridge b
& =
a \unidirectionalbridge b \parallel b \unidirectionalbridge a
\\
\loser a
& =
a \distributor []
\\
\duplicator a
& =
a \distributor [a, a]
\\
\duploser a
& =
\loser a \parallel \duplicator a
\end{align}

\end{definition}

Note that among the derivations in
\autoref{communication-language-embedding} only the one of distributors
directly refers to constructs outside the communication language; all
other derivations refer to communication language constructs only.
Therefore, we consider $\unidirectionalbridge$, $\bidirectionalbridge$,
$\loser$, $\duplicator$, and~$\duploser$ as merely providing syntactic
sugar and discuss only the communication language fragment formed by
$\deadlock$, $\Rightarrow$, $\parallel$, and~$\nu$ in the remainder of
this section.

Since the communication language is embedded in the Þ-calculus, we can
derive a formal semantics for it from the formal semantics of the
Þ-calculus.\footnote{Also this semantics ignores mobility, because we
derive it from the no-mobility version of the Þ-calculus semantics.
However, unlike with the Þ-calculus, the gap between the semantics
presented here and the full semantics is minimal, since the absence of
arbitrary sending makes it impossible to send local channels to the
environment. In fact, the only additional feature of the full semantics
is that it accounts for the possibility of distributors \emph{receiving}
previously unknown channels from the environment.} For dealing with the
constructs inherited from the Þ-calculus, we reuse the corresponding
transition rules, which are $\rulecommunicationltr$,
$\rulecommunicationrtl$, $\ruleparallelleft$, $\ruleparallelright$,
and~$\rulenewchannel$. For dealing with distributors, which are
receivers of a particular shape, we specialize the $\rulereceiving$-rule
appropriately, resulting in a new rule~$\ruledistribution$. Transitions
from distributors with at least one target channel result in processes
that contain senders. Therefore, our transition system must be able to
cope with the additional presence of senders in processes. We reuse the
$\rulesending$-rule from the Þ-calculus for this purpose. We call the
communication language extended with senders the extended communication
language.

\begin{definition}
    [Syntax of processes of the extended communication language]

The syntax of processes of the extended communication language is given
by the following BNF rule, where $a$,~$b$, and~$b_{i}$ denote channels,
$x$~denotes values, $p$~and~$q$ denote processes, and $P$~denotes
functions from channels to processes:
\begin{equation*}
\mathit{Process} \bnfdef
    \deadlock                             \mid
    a \send x                             \mid
    a \distributor [b_{1}, \ldots, b_{n}] \mid
    a \unidirectionalbridge b             \mid
    a \bidirectionalbridge b              \mid
    \loser a                              \mid
    \duplicator a                         \mid
    \duploser a                           \mid
    p \parallel q                         \mid
    \newchannel{a} P \app a
\end{equation*}
The processes generated by the different alternatives of this BNF rule
are called the stop process, senders, distributors, unidirectional
bridges, bidirectional bridges, losers, duplicators, duplosers, parallel
compositions, and restrictions, respectively. Parallel composition is
right-associative and has lowest precedence; the other constructs have
highest precedence.

\end{definition}

\begin{proposition}[Semantics of the extended communication language]

The restriction of the transition relation~$\trans{}$ of the Þ-calculus
to processes of the extended communication language is generated by the
introduction rules in
\autoref{extended-communication-language-no-mobility-transition-rules}.%
\begin{figure}

\centering
\begin{math}
\begin{rulesinfigure}
\sendingrule
\nextline
\introrule
    {}
    {
        a \distributor [b_{1}, \ldots, b_{n}]
        \trans{a \receiving x}
        (
            b_{1} \send x \parallel
            \ldots        \parallel
            b_{n} \send x \parallel
            \deadlock
        ) \parallel
        a \distributor [b_{1}, \ldots, b_{n}]
    }
    {\ruledistribution}
\nextline
\parallelandnewchannelrules
\end{rulesinfigure}
\end{math}

\caption{%
    No-mobility versions
    of the transition rules
    of the extended communication language%
}

\label{extended-communication-language-no-mobility-transition-rules}

\end{figure}

\end{proposition}

\begin{lemma}
    [%
        Strong and weak bisimilarity of processes of the extended
        communication language%
    ]

The strong and weak bisimilarity relations derived from the transition
relation described in
\autoref{extended-communication-language-no-mobility-transition-rules}
can be obtained by restricting the bisimilarity relations
$\sim$~and~$\approx$ of the Þ-calculus to processes of the extended
communication language.

\end{lemma}

\begin{proof}

The only difference between the transition rules in
Figures~\ref{thorn-calculus-no-mobility-transition-rules}~and~%
\ref{extended-communication-language-no-mobility-transition-rules} is
that the former include~$\rulereceiving$ where the latter
include~$\ruledistribution$. However, $\ruledistribution$ is
just~$\rulereceiving$ restricted to those situations where the source
process has the shape of a distributor. Therefore, simulation in the
extended communication language can be performed according to the
Þ-calculus semantics and only in this way. As a result, processes are
strongly or weakly bisimilar according to the semantics of the extended
communication language exactly if they are bisimilar (strongly or
weakly, respectively) according to the semantics of the Þ-calculus.
\end{proof}

\begin{corollary}

Lemmas \ref{strong-weak-bisimilarity-inclusion}~to~%
\ref{fundamental-bisimilarities} carry over to the extended
communication language.

\end{corollary}

\section{A Proof of Idempotency of Repeating Receivers}

\label{repeating-receivers-idempotency-proof}

As mentioned in \autoref{introduction}, the proof of equivalence of
direct broadcast and broadcast via multicast as presented in our
previous work~\cite{jeltsch:2022} relies on certain lower-level
bisimilarity statements. Meanwhile, we have developed proofs for most of
these statements.\footnote{These proofs can be found in the accompanying
Isabelle code.} Some of these proofs merely reduce the respective
bisimilarity statements to more basic bisimilarity statements, but the
proofs of the fundamental statements refer directly to the transition
system semantics of the Þ-calculus and the communication language. These
latter proofs are bisimulation proofs in the style that we advocate in
this paper.

To illustrate this style, let us turn our attention to a group of
idempotency laws. First note that all communication language constructs
not inherited from the Þ-calculus are idempotent up to strong
bisimilarity with respect to parallel composition, which is vital for
our broadcast equivalence proof.

\begin{lemma}[Idempotency of genuine communication language constructs]

The following idempotency properties hold:
\begin{align}
a \distributor bs \parallel a \distributor bs
& \sim
a \distributor bs
\\
a \unidirectionalbridge b \parallel a \unidirectionalbridge b
& \sim
a \unidirectionalbridge b
\\
a \bidirectionalbridge b \parallel a \bidirectionalbridge b
& \sim
a \bidirectionalbridge b
\\
\loser a \parallel \loser a
& \sim
\loser a
\\
\duplicator a \parallel \duplicator a
& \sim
\duplicator a
\\
\duploser a \parallel \duploser a
& \sim
\duploser a
\end{align}

\end{lemma}

The proofs of these idempotency properties are part of the accompanying
Isabelle code. They reduce these properties to a fundamental idempotency
law, which is idempotency of repeating receivers.

\begin{lemma}[Idempotency of repeating receivers]

\label{repeating-receivers-idempotency}

The following idempotency property holds:%
\begin{equation}
a \repeatedreceive{x} P \app x \parallel a \repeatedreceive{x} P \app x
\sim
a \repeatedreceive{x} P \app x
\end{equation}

\end{lemma}

It is this idempotency law that we use as our example for demonstrating
our style of bisimulation proofs that are concise and human-friendly yet
machine-checked. However, before we turn to the Isabelle/HOL proof that
exhibits this style, we provide a semi-formal proof of this law.

\begin{proof}
    [Semi-formal proof of \autoref{repeating-receivers-idempotency}]

We prove the idempotency of repeating receivers by bisimulation up to
strong bisimilarity and context.
\begin{description}

\item[Forward simulation.]

Assume that $a \repeatedreceive{x} P \app x \parallel a
\repeatedreceive{x} P \app x \trans{\alpha} s$ for arbitrary but fixed
$\alpha$~and~$s$. Looking at
\autoref{thorn-calculus-no-mobility-transition-rules}, we can see that
only rules $\rulecommunicationltr$,~$\rulecommunicationrtl$,
$\ruleparallelleft$, and~$\ruleparallelright$ can in principle introduce
this transition, given that its source process is a parallel
composition.
\begin{description}

\item[Rules $\rulecommunicationltr$~and~$\rulecommunicationrtl$.]

Introducing the above transition using either of these rules requires a
$\sending$-transition from $a \repeatedreceive{x} P \app x$, which is
not possible according to \autoref{repeating-receivers-transitions}.

\item[Rule~$\ruleparallelleft$.]

For introducing the above transition using this rule, there must be a
process~$q$ such that the following statements hold:
\begin{gather}
\label{transition-from-left-repeating-receiver}
a \repeatedreceive{x} P \app x \trans{\alpha} q
\tag{i}
\\
\label{both-repeating-receivers-target-process}
s = q \parallel a \repeatedreceive{x} P \app x
\tag{ii}
\end{gather}
Based on \autoref{repeating-receivers-transitions},
statement~\statementref{transition-from-left-repeating-receiver} implies
that there is an~$x$ for which the following propositions are true:
\begin{gather}
\label{repeating-receiver-action}
\alpha = a \receiving x
\tag{iii}
\\
\label{left-repeating-receiver-target-process}
q = P \app x \parallel a \repeatedreceive{x} P \app x
\tag{iv}
\end{gather}
From
\statementref{both-repeating-receivers-target-process}~and~\statementref{left-repeating-receiver-target-process},
we can deduce the following:
\begin{gather}
\label{both-repeating-receivers-concrete-target-process}
s
=
(P \app x \parallel a \repeatedreceive{x} P \app x)
\parallel
a \repeatedreceive{x} P \app x
\tag{v}
\end{gather}
Because of
\statementref{repeating-receiver-action}~and~\statementref{both-repeating-receivers-concrete-target-process},
the transition we have started with has the following concrete shape:
\begin{equation}
\label{original-transition-explicitly}
a \repeatedreceive{x} P \app x \parallel a \repeatedreceive{x} P \app x
\trans{a \receiving x}
(P \app x \parallel a \repeatedreceive{x} P \app x)
\parallel
a \repeatedreceive{x} P \app x
\tag{vi}
\end{equation}
We simulate this transition with the transition $a \repeatedreceive{x} P
\app x \trans{a \receiving x} P \app x \parallel a \repeatedreceive{x} P
\app x$, whose existence follows from
\statementref{transition-from-left-repeating-receiver},~\statementref{repeating-receiver-action},
and~\statementref{left-repeating-receiver-target-process}. The target
processes of these two transitions are the respective source processes
up to strong bisimilarity and context. To see why, observe that the
first target process can be transformed into a bisimilar one as follows,
employing
\begingroup\renewcommand{\equationautorefname}{Bisimilarity}\autoref{parallel-associativity}\endgroup:
\begin{equation}
(P \app x \parallel a \repeatedreceive{x} P \app x)
\parallel
a \repeatedreceive{x} P \app x
\sim
P \app x
\parallel
(
    a \repeatedreceive{x} P \app x
    \parallel
    a \repeatedreceive{x} P \app x
)
\tag{vii}
\end{equation}
Removing the common context $P \app x \parallel [\cdot]$ from the result
of this transformation and the target process of the simulating
transition yields $a \repeatedreceive{x} P \app x \parallel a
\repeatedreceive{x} P \app x$ and $a \repeatedreceive{x} P \app x$.

\item[Rule~$\ruleparallelright$.]

This rule can be handled analogously to rule~$\ruleparallelleft$.

\end{description}

\item[Backward simulation.]

Assume that $a \repeatedreceive{x} P \app x \trans{\alpha} s$ for
arbitrary but fixed $\alpha$~and~$s$.
\autoref{repeating-receivers-transitions} tells us that there is an~$x$
such that $\alpha = a \receiving x$ and $s = P \app x \parallel a
\repeatedreceive{x} P \app x$, from which we can deduce that the
transition we have started with is concretely $a \repeatedreceive{x} P
\app x \trans{a \receiving x} P \app x \parallel a \repeatedreceive{x} P
\app x$. By applying rule~$\ruleparallelleft$, we can turn this
transition into
transition~\statementref{original-transition-explicitly}, which we use
as the simulating transition. The target processes of the original and
the simulating transition are the respective source processes up to
strong bisimilarity and context, for essentially the same reasons as in
the case of forward simulation of transitions generated by
rule~$\ruleparallelleft$.\qedhere

\end{description}

\end{proof}

\autoref{repeating-receivers-idempotency-formal-proof} presents the
formal proof of \autoref{repeating-receivers-idempotency}, which is
conducted in Isabelle/HOL. To not bother the reader with technicalities,
this presentation omits the subproofs that justify the atomic reasoning
steps. These subproofs are only short, straightforward applications of
lemmas and proof methods. The complete proof can be found in the
accompanying Isabelle code. Note that the formal proof, also as shown
here, refers to the full semantics and thus has to deal with mobility.
\begin{figure}

\begin{lstlisting}[language=isabelle,basicstyle=\small]
lemma repeated_receive_idempotency:
  shows $a \repeatedreceive{x} P \app x \parallel a \repeatedreceive{x} P \app x \sim a \repeatedreceive{x} P \app x$
proof (coinduction rule: up_to_rule [where $\mathcal{F} = \constantbisimilarity \chain \mutantlifting$])
  case (forward_simulation $\alpha$ $s$)
  then show $\schematic{case}$
  proof cases
    case (parallel_left_io $\eta$ $a'$ $n$ $x$ $q$)
    from $\fact{a \repeatedreceive{x} P \app x \trans{\IO{\eta}{a'}{n}{x}} q}$ obtain $t$ where $q = t \parallel (a \repeatedreceive{x} P \app x) \suffixadapted n$
      $\proofplaceholder$
    with $\fact{a \repeatedreceive{x} P \app x \trans{\IO{\eta}{a'}{n}{x}} q}$ have $a \repeatedreceive{x} P \app x \trans{\IO{\eta}{a'}{n}{x}} t \parallel (a \repeatedreceive{x} P \app x) \suffixadapted n$
      $\proofplaceholder$
    moreover have $(t \parallel r \suffixadapted n) \parallel r \suffixadapted n \sim t \parallel (r \parallel r) \suffixadapted n$ for $r$
      $\proofplaceholder$
    ultimately show $\schematic{thesis}$
      $\proofplaceholder$
  next
    case (parallel_right_io $\eta$ $a'$ $n$ $x$ $q$)
    from $\fact{a \repeatedreceive{x} P \app x \trans{\IO{\eta}{a'}{n}{x}} q}$ obtain $t$ where $q = t \parallel (a \repeatedreceive{x} P \app x) \suffixadapted n$
      $\proofplaceholder$
    with $\fact{a \repeatedreceive{x} P \app x \trans{\IO{\eta}{a'}{n}{x}} q}$ have $a \repeatedreceive{x} P \app x \trans{\IO{\eta}{a'}{n}{x}} t \parallel (a \repeatedreceive{x} P \app x) \suffixadapted n$
      $\proofplaceholder$
    moreover have $r \suffixadapted n \parallel (t \parallel r \suffixadapted n) \sim t \parallel (r \parallel r) \suffixadapted n$ for $r$
      $\proofplaceholder$
    ultimately show $\schematic{thesis}$
      $\proofplaceholder$
  qed  (blast elim: transition_$\;$from_repeated_receive)+
next
  case (backward_simulation $\alpha$ $s$)
  from $\fact{a \repeatedreceive{x} P \app x \trans{\alpha} s}$ obtain $n$ and $x$ where $\alpha = a \realreceiving{n} x$ and $s = \postreceive{n}{x}{P} \parallel (a \repeatedreceive{x} P \app x) \suffixadapted n$
    $\proofplaceholder$
  with $\fact{a \repeatedreceive{x} P \app x \trans{\alpha} s}$ have $a \repeatedreceive{x} P \app x \trans{a \realreceiving{n} x} \postreceive{n}{x}{P} \parallel (a \repeatedreceive{x} P \app x) \suffixadapted n$
    $\proofplaceholder$
  then have $a \repeatedreceive{x} P \app x \parallel a \repeatedreceive{x} P \app x \trans{a \realreceiving{n} x} (\postreceive{n}{x}{P} \parallel (a \repeatedreceive{x} P \app x) \suffixadapted n) \parallel (a \repeatedreceive{x} P \app x) \suffixadapted n$
    $\proofplaceholder$
  moreover have $(t \parallel r \suffixadapted n) \parallel r \suffixadapted n \sim t \parallel (r \parallel r) \suffixadapted n$ for $r$
    $\proofplaceholder$
  ultimately show $\schematic{case}$
    $\proofplaceholder$
qed respectful
\end{lstlisting}

\caption{Formal proof of idempotency of repeating receivers}

\label{repeating-receivers-idempotency-formal-proof}

\end{figure}

To aid understanding of the formal proof, let us point out a few things:
\begin{itemize}

\item

The initial proof method uses the term $\constantbisimilarity \chain
\mutantlifting$ to specify “up to strong bisimilarity and context” as
the “up to” method to use. To guarantee that the provided term specifies
an “up to” method that is sound, we have to prove that it fulfills a
certain condition. We do that by invoking the automated proof method
\lstinline|respectful| at the end of the proof.

\item

The part on forward simulation mentions actions of the form
$\IO{\eta}{a}{n}{x}$. Such actions can be sending or receiving actions.
For reasons having to do with mobility, there are separate versions of
$\ruleparallelleft$~and~$\ruleparallelright$ for sending and receiving
actions on the one hand and the internal-transfer action on the other.
The cases \lstinline|parallel_left_io| and \lstinline|parallel_right_io|
are only about sending and receiving, not about internal transfer.

\item

There are no explicit proof steps for showing that the original
transition of a forward simulation cannot be introduced using
$\rulecommunicationltr$~or~$\rulecommunicationrtl$. We have that
automatically shown by the proof method
\lstinline|(blast elim: transition_$\;$from_repeated_receive)+| at the
end of the forward simulation part. This proof method additionally shows
that said transition cannot be introduced using the internal-transfer
versions of $\ruleparallelleft$~and~$\ruleparallelright$ mentioned in
the previous item.

\item

Mobility makes it possible to receive previously unknown channels from
the environment. To deal with this possibility, some tweaks are
necessary, namely adding $\nobinarg \suffixadapted n$ in a few places,
switching to a more powerful kind of receiving action,
$\realreceiving{n}$, and replacing $P \app x$ by
$\postreceive{n}{x}{P}$. A deeper discussion of these tweaks would be
outside the scope of this paper.

\end{itemize}

As can be seen, the formal proof is quite similar to the semi-formal
one, which we consider a strength of our work. It is generally more
compact, but the handling of forward simulation of transitions generated
by rule~$\ruleparallelright$ had to be spelled out, where the
semi-formal proof could just state that it is analogous to what was done
for rule~$\ruleparallelleft$.

\section{Bisimulation Proofs for Humans and Machines}

The semi-formal proof of \autoref{repeating-receivers-idempotency} is
geared toward human readers, and its style has been chosen accordingly.
The formal proof, by following the semi-formal proof rather closely,
retains this human-friendly style to a large extent but is
machine-checked at the same time. This achievement rests on the
combination of several tools:
\begin{description}

\item[The Isabelle/Isar proof language.]

Isabelle/Isar~\cite{wenzel:2022} is a structured, declarative proof
language that incorporates elements of mathematical prose. With these
characteristics, Isar proofs differ notably from proof terms as well as
tactics-based proof scripts, with the result of being better
understandable by humans. Despite its human-friendliness, Isar comes
with a precise semantics, and the correctness of Isar proofs can be
checked using the Isabelle proof assistant.

The use of Isar is crucial for having the formal proof largely reflect
the semi-formal proof. The block structure achieved by employing
\lstinline|proof|, \lstinline|case|, \lstinline|next|, and
\lstinline|qed| resembles the overall structure of the semi-formal
proof, in particular the nesting of subproofs and the distinction
between forward and backward simulation as well as between different
introduction rules. At the bottom layer, intermediate facts are
explicitly stated and later accessed using Isar’s flexible means for
fact referencing. Other, minor, features of Isar serve to further narrow
the gap between the formal and the semi-formal proof.

\item[A formalized algebra of “up to” methods.]

Both the semi-formal and the formal proof have to cope with the fact
that transitions from $a \repeatedreceive{x} P \app x \parallel a
\repeatedreceive{x} P \app x$ and $a \repeatedreceive{x} P \app x$ do
not result in these processes again but only in processes that can be
derived from them by adding a common context and performing a
bisimilarity transformation. However, this is not a problem, because
employing the “up to strong bisimilarity and context” method bridges
this gap.

A bisimulation proof that does not employ “up to” methods would be much
more complex. Such a proof would have to show that bisimulation is also
possible for the above-mentioned target processes and recursively for
any processes that arise from bisimulation of previously considered
processes. In the end, instead of dealing with repeating receivers only,
the proof would have to deal with all processes of the form $u_{1}
\parallel \ldots \parallel u_{n} \parallel a \repeatedreceive{x} P \app
x$. Since the processes to be proved bisimilar contain a total of three
repeating receivers, this would result in an enormous amount of
boilerplate that would obscure the key arguments of the proof.
Furthermore, such a proof would be hard to develop in the first place.

In order to prevent such issues, we have implemented an algebra of “up
to” methods that are guaranteed to be sound, using Isabelle/HOL. This
implementation enables developers of formal bisimulation proofs to
construct custom “up to” methods that fit the specific bisimilarity
statements to prove. In the proof of idempotency of repeating receivers,
we use the “up to” method $\constantbisimilarity \chain \mutantlifting$.
This method is built from the primitive methods
$\mutantlifting$~and~$\constantbisimilarity$. $\mutantlifting$~requires
target processes to be source processes up to context\footnote{Actually
up to mutation, which is more general than up to context.}, and
$\constantbisimilarity$ requires target processes to be strongly
bisimilar, independently of source processes. The operator~$\chain$
serves to combine the two. Note that $\constantbisimilarity \chain
\mutantlifting$ allows only the first process to deviate by strong
bisimilarity, which is the one for which we need this possibility; full
“up to strong bisimilarity and context” is denoted by
$\constantbisimilarity \chain \mutantlifting \chain
\constantbisimilarity$.

\item[The \textit{coinduction} proof method.]

Isabelle’s \lstinline|coinduction| proof method~\cite{blanchette:2014}
makes it possible to conduct coinductive proofs using the
\lstinline|proof|–\lstinline|case|–\lstinline|next|–\lstinline|qed|
style exemplified by our formal proof of
\autoref{repeating-receivers-idempotency}. Isabelle/HOL supports
coinductive definitions of data types and predicates, and in its default
mode the \lstinline|coinduction| method enables reasoning along the
coinductive structure of the data types and predicates so defined. In
the case of bisimilarity, which is a coinductively defined predicate,
this leads to plain bisimulation proofs, those that do not employ “up
to” methods.

However, the \lstinline|coinduction| method can also work with
user-provided coinduction rules, which can be lemmas derived from the
coinduction rules induced by coinductive data type and predicate
definitions. This allows us to use the \lstinline|coinduction| method
for bisimulation proofs that apply “up to” methods. For employing a
concrete “up to” method, we can instantiate the generic lemma
\lstinline|up_to_rule| for this “up to” method and provide the resulting
fact as the coinduction rule to use to the \lstinline|coinduction| proof
method.

A feature of the \lstinline|coinduction| method that helps making proofs
concise is the automatic derivation of bisimulation relations. As
indicated in the previous item, bisimulation relations often have to
cover more than just the processes to be proved bisimilar if “up to”
methods are not used, since target processes typically deviate from
source processes. However, in most bisimulation proofs that do use “up
to” methods, including our proof of idempotency of repeating receivers,
this issue does not arise, and the bisimulation relation of choice is
the one that just covers the processes whose bisimilarity is to be
shown. The \lstinline|coinduction| method derives this relation from the
proof goal and automatically shows the trivial statement that the
processes to be proved bisimilar are in this relation.\footnote{This is
what distinguishes it from the \lstinline|coinduct|
method~\cite[Subsection~6.5.2]{wenzel:2022}, which requires the user to
specify the bisimulation relation and prove that the processes to be
proved bisimilar are in this relation.} As a result, the proof can
concentrate on the actual bisimulation.

\item[Higher-order abstract syntax.]

Higher-order abstract syntax (HOAS)~\cite{pfenning:1988} is a technique
of embedding an object language in a higher-order host language where
name binding in the object language is expressed using functions of the
host language. When not using HOAS, formal proofs that involve binding
constructs tend to be littered with boilerplate for dealing with issues
like name capturing and freshness conditions. By employing HOAS, this
problem can be prevented. As we have seen in
\autoref{repeating-receivers-idempotency-proof}, also our use of HOAS
necessitates additional handling of technicalities as soon as mobility
is taken into account. However, the corresponding amount of extra code
tends to be low compared to the amount of extra code necessary with
non-HOAS approaches, including those that make use of nominal logic,
which is generally boilerplate-reducing.\footnote{For example, the
complete implementation of “up to” methods for the Þ-calculus is less
than half the size of the implementation of “up to” methods for
$\psi$-calculi~\cite{pohjola:2016}, which also uses Isabelle/HOL and
employs the Nominal Isabelle framework~\cite{urban:2008}. This
considerable difference in code size may also be due to $\psi$-calculi
explicitly handling computation and conditional execution, but the
reason that the Þ-calculus does not have to explicitly deal with these
features is also because of its use of HOAS.}

The use of HOAS makes it possible to construct exotic terms, that is,
terms where functions representing binding yield subterms whose
structure depends on the arguments of these functions. In the case of
receivers, we consider this a feature, as it allows us to handle
computation and conditional execution within the meta-language. However,
in the case of restrictions, exotic terms may become an issue when
treating mobility naïvely; for example, Bisimilarities
\ref{parallel-compatibility-strong}~and~\ref{parallel-compatibility-weak}
may not hold anymore. The typical solution to such problems is to
restrict the calculus in question to terms that are not exotic. The
solution of the Þ-calculus, however, is different: exotic terms can be
constructed freely, but the transition system semantics does not allow
transitions from exotic restrictions. Only the full transition system,
which can be found in the accompanying Isabelle code, has this feature
of preventing such transitions. To achieve it, the transition system has
to maintain lists of channels introduced by restrictions, and this
results in the need for the additional tweaks present in the formal
proof.

\end{description}

\section{Related Work}

Various domain-specific languages for modeling communication networks
and reasoning about them are discussed in the literature. One of them is
NetKAT~\cite{anderson:2014}, a network programming language based on
Kleene algebra with tests (KAT) that features a complete deductive
system and a PSPACE decision procedure. Unlike our communication
language, NetKAT lacks restriction and, being a sequential language,
also parallel composition. On the other hand, it allows for packet
inspection and modification. Another example of a network communication
DSL is Nettle~\cite{voellmy:2011}, a language for programming OpenFlow
networks that is embedded in Haskell and based on the principles of
functional reactive programming. Like with NetKAT, packet inspection and
modification is also possible with Nettle.

Process calculi for describing and verifying communication networks have
been an active area of research. For example, the
$\omega$-calculus~\cite{singh:2010} is a process calculus devised to
formally reason about mobile ad-hoc networks (MANETs). It is a
conservative extension of the $\pi$-calculus that has built-in support
for unicast and broadcast communication as well as location-based
scoping. We have designed the Þ-calculus as a general-purpose process
calculus and have thus avoided the inclusion of application-specific
features like support for broadcast communication. That said, such
features can be implemented on top of the Þ-calculus, as the definition
of the communication language as an embedded DSL and Examples
\ref{reliable-anycast}~and~\ref{unreliable-broadcast} show.

Several well-known process calculi have been formalized by Bengtson and
colleagues in Isabelle/HOL, in particular the
$\pi$-calculus~\cite{bengtson:2009} and
$\psi$-calculi~\cite{bengtson:2016}. Unlike our formalization of the
Þ-calculus, those formalizations do not use HOAS but Nominal
Isabelle~\cite{urban:2008} for dealing with name binding. It appears
that this makes them more complex than the Þ-calculus formalization,
although one has to consider that they use version~1 of Nominal
Isabelle, not the improved version~2. Furthermore, the formalizations by
Bengtson et~al.\ suffer from considerable repetition, for example in
their handling of strong and weak bisimilarity. The Þ-calculus
formalization, on the other hand, makes more use of abstractions and
achieves more code reuse this way. Another, albeit minor, advantage of
the Þ-calculus formalization is its use of the \lstinline|coinduction|
proof method. The above-mentioned formalizations of the $\pi$-calculus
and $\psi$-calculi use the less powerful \lstinline|coinduct| method,
resulting in more boilerplate code. Finally, “up to” methods help to
avoid repetitive, technical proof code on a large scale, which we
leverage in the Þ-calculus formalization and the developments built on
it, using our formalized algebra of “up to” methods. The formalizations
by Bengtson et~al.\ also make use of “up to” methods, but the authors
have only proved the soundness of a few specific methods. That said,
Åman Pohjola and Parrow have developed a framework for “up to” methods
for $\psi$-calculi~\cite{pohjola:2016}.

\section{Conclusion}

We have presented a transition system semantics for the Þ-calculus,
which is a general-purpose process calculus embedded in Isabelle/HOL,
and derived from it a transition system semantics for a custom network
communication language, which is embedded in the Þ-calculus. Building on
this foundation and based on an example related to network
communication, we have shown a way of conducting bisimulation proofs
such that they become concise and human-friendly, while being
machine-checked at the same time. Our proving style stems from combining
the Isabelle/Isar proof language, an algebra of “up to” methods
formalized in Isabelle/HOL, Isabelle’s \lstinline|coinduction| proof
method, and higher-order abstract syntax.

\section{Ongoing and Future Work}

As mentioned in \autoref{repeating-receivers-idempotency-proof}, we have
proved most of the lower-level bisimilarity statements on which our
broadcast equivalence proof~\cite{jeltsch:2022} relies. At the moment,
we are completing the last proofs of such statements.

In accordance with our research program mentioned in
\autoref{introduction}, we plan to verify further design refinement
steps that the consensus protocols of the Ouroboros family have
undergone. The refinement step we want to tackle next is the replacement
of whole-chain distribution with a protocol for updating chains
incrementally. Furthermore, we want to add some missing bits, in
particular documentation, to the formalization of the Þ-calculus and the
algebra of “up to” methods and submit both formalizations to Isabelle’s
Archive of Formal Proofs (AFP)\footnote{See
\url{https://www.isa-afp.org/}.}.

\section*{Acknowledgements}

This work was funded by Input Output. We are thankful to Input Output
for giving us the opportunity to work on numerous interesting topics,
including the one described in this paper. Furthermore, we want to thank
the anonymous reviewers and James Chapman for their various suggestions
for improvement of this paper.

\bibliography{paper}

\begin{thebibliography}{10}
\providecommand{\bibitemdeclare}[2]{}
\providecommand{\surnamestart}{}
\providecommand{\surnameend}{}
\providecommand{\urlprefix}{Available at }
\providecommand{\url}[1]{\texttt{#1}}
\providecommand{\href}[2]{\texttt{#2}}
\providecommand{\urlalt}[2]{\href{#1}{#2}}
\providecommand{\doi}[1]{doi:\urlalt{https://doi.org/#1}{#1}}
\providecommand{\eprint}[1]{arXiv:\urlalt{https://arxiv.org/abs/#1}{#1}}
\providecommand{\bibinfo}[2]{#2}

\bibitemdeclare{inproceedings}{anderson:2014}
\bibitem{anderson:2014}
\bibinfo{author}{Carolyn~Jane \surnamestart Anderson\surnameend},
  \bibinfo{author}{Nate \surnamestart Foster\surnameend},
  \bibinfo{author}{Arjun \surnamestart Guha\surnameend},
  \bibinfo{author}{Jean-Baptiste \surnamestart Jeannin\surnameend},
  \bibinfo{author}{Dexter \surnamestart Kozen\surnameend},
  \bibinfo{author}{Cole \surnamestart Schlesinger\surnameend} \&
  \bibinfo{author}{David \surnamestart Walker\surnameend}
  (\bibinfo{year}{2014}): \emph{\bibinfo{title}{{NetKAT}: Semantic Foundations
  for Networks}}.
\newblock In: {\slshape \bibinfo{booktitle}{Proceedings of the 41st {ACM
  SIGPLAN–SIGACT} Symposium on Principles of Programming Languages}},
  \bibinfo{publisher}{ACM}, \bibinfo{address}{New York}, pp.
  \bibinfo{pages}{113--126}, \doi{10.1145/2535838.2535862}.

\bibitemdeclare{inproceedings}{badertscher:2018}
\bibitem{badertscher:2018}
\bibinfo{author}{Christian \surnamestart Badertscher\surnameend},
  \bibinfo{author}{Peter \surnamestart Gaži\surnameend},
  \bibinfo{author}{Aggelos \surnamestart Kiayias\surnameend},
  \bibinfo{author}{Alexander \surnamestart Russell\surnameend} \&
  \bibinfo{author}{Vassilis \surnamestart Zikas\surnameend}
  (\bibinfo{year}{2018}): \emph{\bibinfo{title}{{Ouroboros Genesis}: Composable
  Proof-of-Stake Blockchains with Dynamic Availability}}.
\newblock In: {\slshape \bibinfo{booktitle}{Proceedings of the 2018 {ACM
  SIGSAC} Conference on Computer and Communications Security}},
  \bibinfo{publisher}{ACM}, \bibinfo{address}{New York}, pp.
  \bibinfo{pages}{913--930}, \doi{10.1145/3243734.3243848}.

\bibitemdeclare{article}{bengtson:2009}
\bibitem{bengtson:2009}
\bibinfo{author}{Jesper \surnamestart Bengtson\surnameend} \&
  \bibinfo{author}{Joachim \surnamestart Parrow\surnameend}
  (\bibinfo{year}{2009}): \emph{\bibinfo{title}{Formalising the $\pi$-Calculus
  Using Nominal Logic}}.
\newblock {\slshape \bibinfo{journal}{Logical Methods in Computer Science}}
  \bibinfo{volume}{5}(\bibinfo{number}{2}), pp. \bibinfo{pages}{1--36},
  \doi{10.2168/LMCS-5(2:16)2009}.

\bibitemdeclare{article}{bengtson:2016}
\bibitem{bengtson:2016}
\bibinfo{author}{Jesper \surnamestart Bengtson\surnameend},
  \bibinfo{author}{Joachim \surnamestart Parrow\surnameend} \&
  \bibinfo{author}{Tjark \surnamestart Weber\surnameend}
  (\bibinfo{year}{2016}): \emph{\bibinfo{title}{Psi-Calculi in {Isabelle}}}.
\newblock {\slshape \bibinfo{journal}{Journal of Automated Reasoning}}
  \bibinfo{volume}{56}(\bibinfo{number}{1}), pp. \bibinfo{pages}{1--47},
  \doi{10.1007/s10817-015-9336-2}.

\bibitemdeclare{incollection}{blanchette:2014}
\bibitem{blanchette:2014}
\bibinfo{author}{Jasmin~Christian \surnamestart Blanchette\surnameend},
  \bibinfo{author}{Johannes \surnamestart Hölzl\surnameend},
  \bibinfo{author}{Andreas \surnamestart Lochbihler\surnameend},
  \bibinfo{author}{Lorenz \surnamestart Panny\surnameend},
  \bibinfo{author}{Andrei \surnamestart Popescu\surnameend} \&
  \bibinfo{author}{Dmitriy \surnamestart Traytel\surnameend}
  (\bibinfo{year}{2014}): \emph{\bibinfo{title}{Truly Modular (Co)datatypes for
  {Isabelle/HOL}}}.
\newblock In \bibinfo{editor}{Gerwin \surnamestart Klein\surnameend} \&
  \bibinfo{editor}{Ruben \surnamestart Gamboa\surnameend}, editors: {\slshape
  \bibinfo{booktitle}{Interactive Theorem Proving}}, {\slshape
  \bibinfo{series}{Lecture Notes in Computer Science}} \bibinfo{volume}{8558},
  \bibinfo{publisher}{Springer}, \bibinfo{address}{Berlin/Heidelberg, Germany},
  pp. \bibinfo{pages}{93--110}, \doi{10.1007/978-3-319-08970-6\_7}.

\bibitemdeclare{incollection}{david:2018}
\bibitem{david:2018}
\bibinfo{author}{Bernardo \surnamestart David\surnameend},
  \bibinfo{author}{Peter \surnamestart Gaži\surnameend},
  \bibinfo{author}{Aggelos \surnamestart Kiayias\surnameend} \&
  \bibinfo{author}{Alexander \surnamestart Russell\surnameend}
  (\bibinfo{year}{2018}): \emph{\bibinfo{title}{{Ouroboros Praos}: An
  Adaptively-Secure, Semi-Synchronous Proof-of-Stake Blockchain}}.
\newblock In \bibinfo{editor}{Jesper \surnamestart Buus~Nielsen\surnameend} \&
  \bibinfo{editor}{Vincent \surnamestart Rijmen\surnameend}, editors: {\slshape
  \bibinfo{booktitle}{Advances in Cryptology -- EUROCRYPT~2018}}, {\slshape
  \bibinfo{series}{Lecture Notes in Computer Science}} \bibinfo{volume}{10821},
  \bibinfo{publisher}{Springer}, \bibinfo{address}{Berlin/Heidelberg, Germany},
  pp. \bibinfo{pages}{66--98}, \doi{10.1007/978-3-319-78375-8\_3}.

\bibitemdeclare{incollection}{honda:1991}
\bibitem{honda:1991}
\bibinfo{author}{Kohei \surnamestart Honda\surnameend} \&
  \bibinfo{author}{Mario \surnamestart Tokoro\surnameend}
  (\bibinfo{year}{1991}): \emph{\bibinfo{title}{An Object Calculus for
  Asynchronous Communication}}.
\newblock In \bibinfo{editor}{Pierre \surnamestart America\surnameend}, editor:
  {\slshape \bibinfo{booktitle}{ECOOP~'91 European Conference on
  Object-Oriented Programming}}, {\slshape \bibinfo{series}{Lecture Notes in
  Computer Science}} \bibinfo{volume}{512}, \bibinfo{publisher}{Springer},
  \bibinfo{address}{Berlin/Heidelberg, Germany}, pp. \bibinfo{pages}{133--147},
  \doi{10.1007/BFb0057019}.

\bibitemdeclare{article}{jeltsch:2022}
\bibitem{jeltsch:2022}
\bibinfo{author}{Wolfgang \surnamestart Jeltsch\surnameend} \&
  \bibinfo{author}{Javier \surnamestart Díaz\surnameend}
  (\bibinfo{year}{2022}): \emph{\bibinfo{title}{Correctness of Broadcast via
  Multicast: Graphically and Formally}}.
\newblock {\slshape \bibinfo{journal}{Electronic Proceedings in Theoretical
  Computer Science}} \bibinfo{volume}{369}, pp. \bibinfo{pages}{37--50},
  \doi{10.4204/EPTCS.369.3}.

\bibitemdeclare{incollection}{kiayias:2017}
\bibitem{kiayias:2017}
\bibinfo{author}{Aggelos \surnamestart Kiayias\surnameend},
  \bibinfo{author}{Alexander \surnamestart Russell\surnameend},
  \bibinfo{author}{Bernardo \surnamestart David\surnameend} \&
  \bibinfo{author}{Roman \surnamestart Oliynykov\surnameend}
  (\bibinfo{year}{2017}): \emph{\bibinfo{title}{{Ouroboros}: A Provably Secure
  Proof-of-Stake Blockchain Protocol}}.
\newblock In \bibinfo{editor}{Jonathan \surnamestart Katz\surnameend} \&
  \bibinfo{editor}{Hovav \surnamestart Shacham\surnameend}, editors: {\slshape
  \bibinfo{booktitle}{Advances in Cryptology -- CRYPTO~2017}}, {\slshape
  \bibinfo{series}{Lecture Notes in Computer Science}} \bibinfo{volume}{10401},
  \bibinfo{publisher}{Springer}, \bibinfo{address}{Berlin/Heidelberg, Germany},
  pp. \bibinfo{pages}{357--388}, \doi{10.1007/978-3-319-63688-7\_12}.

\bibitemdeclare{inproceedings}{pfenning:1988}
\bibitem{pfenning:1988}
\bibinfo{author}{Frank \surnamestart Pfenning\surnameend} \&
  \bibinfo{author}{Conal \surnamestart Elliott\surnameend}
  (\bibinfo{year}{1988}): \emph{\bibinfo{title}{Higher-Order Abstract Syntax}}.
\newblock In: {\slshape \bibinfo{booktitle}{Proceedings of the {ACM SIGPLAN}
  1988 Conference on Programming Language Design and Implementation}},
  \bibinfo{publisher}{ACM}, \bibinfo{address}{New York}, pp.
  \bibinfo{pages}{199--208}, \doi{10.1145/53990.54010}.

\bibitemdeclare{inproceedings}{pohjola:2016}
\bibitem{pohjola:2016}
\bibinfo{author}{Johannes \surnamestart Åman Pohjola\surnameend} \&
  \bibinfo{author}{Joachim \surnamestart Parrow\surnameend}
  (\bibinfo{year}{2016}): \emph{\bibinfo{title}{Bisimulation Up-to Techniques
  for Psi-Calculi}}.
\newblock In: {\slshape \bibinfo{booktitle}{Proceedings of the 5th {ACM
  SIGPLAN} Conference on Certified Programs and Proofs}},
  \bibinfo{publisher}{ACM}, \bibinfo{address}{New York}, pp.
  \bibinfo{pages}{142--153}, \doi{10.1145/2854065.2854080}.

\bibitemdeclare{article}{singh:2010}
\bibitem{singh:2010}
\bibinfo{author}{Anu \surnamestart Singh\surnameend}, \bibinfo{author}{C.~R.
  \surnamestart Ramakrishnan\surnameend} \& \bibinfo{author}{Scott~A.
  \surnamestart Smolka\surnameend} (\bibinfo{year}{2010}):
  \emph{\bibinfo{title}{A Process Calculus for {Mobile Ad Hoc Networks}}}.
\newblock {\slshape \bibinfo{journal}{Science of Computer Programming}}
  \bibinfo{volume}{75}(\bibinfo{number}{6}), pp. \bibinfo{pages}{440--469},
  \doi{10.1016/j.scico.2009.07.008}.

\bibitemdeclare{article}{urban:2008}
\bibitem{urban:2008}
\bibinfo{author}{Christian \surnamestart Urban\surnameend}
  (\bibinfo{year}{2008}): \emph{\bibinfo{title}{Nominal Techniques in
  {Isabelle/HOL}}}.
\newblock {\slshape \bibinfo{journal}{Journal of Automated Reasoning}}
  \bibinfo{volume}{40}(\bibinfo{number}{4}), pp. \bibinfo{pages}{327--356},
  \doi{10.1007/s10817-008-9097-2}.

\bibitemdeclare{incollection}{voellmy:2011}
\bibitem{voellmy:2011}
\bibinfo{author}{Andreas \surnamestart Voellmy\surnameend} \&
  \bibinfo{author}{Paul \surnamestart Hudak\surnameend} (\bibinfo{year}{2011}):
  \emph{\bibinfo{title}{Nettle: Taking the Sting Out of Programming Network
  Routers}}.
\newblock In \bibinfo{editor}{Ricardo \surnamestart Rocha\surnameend} \&
  \bibinfo{editor}{John \surnamestart Launchbury\surnameend}, editors:
  {\slshape \bibinfo{booktitle}{Practical Aspects of Declarative Languages}},
  {\slshape \bibinfo{series}{Lecture Notes in Computer Science}}
  \bibinfo{volume}{6539}, \bibinfo{publisher}{Springer},
  \bibinfo{address}{Berlin/Heidelberg, Germany}, pp. \bibinfo{pages}{235--249},
  \doi{10.1007/978-3-642-18378-2\_19}.

\bibitemdeclare{misc}{wenzel:2022}
\bibitem{wenzel:2022}
\bibinfo{author}{Makarius \surnamestart Wenzel\surnameend}
  (\bibinfo{year}{2022}): \emph{\bibinfo{title}{The {Isabelle/Isar} Reference
  Manual}}.
\newblock
  \bibinfo{howpublished}{\url{https://isabelle.in.tum.de/dist/Isabelle2022/doc/isar-ref.pdf}}.

\end{thebibliography}

\end{document}